\documentclass[11pt]{scrartcl}

\usepackage[utf8]{inputenc}
\usepackage[T1]{fontenc}

\usepackage[bottom=1in,top=0.85in,left=0.85in,right=0.85in]{geometry}\usepackage{times}
\RedeclareSectionCommand[indent=0pt]{subparagraph}

\usepackage{amsmath,amsthm,amssymb}
\usepackage{mathtools,subfigure,wrapfig}
\usepackage[hyphens]{url}
\usepackage{hyperref}  
\usepackage{xspace}
\usepackage{thm-restate}

\usepackage{tabularx}
\newcolumntype{L}{>{\raggedright\arraybackslash}X}
\usepackage[ruled]{algorithm2e}
\usepackage{enumitem}

\usepackage{booktabs}
\newtheorem{theorem}{Theorem}

\newtheorem{lemma}[theorem]{Lemma}

\newtheorem{cor}[theorem]{Corollary}

\usepackage{tikz}
\usetikzlibrary{calc,shapes}

\tikzset{
	point/.style={draw, circle, fill=black, inner sep=0mm, minimum size=6pt},
	rp/.style={point, red!70!yellow},
	bp/.style={point, blue!50!black},
	cluster/.style={draw, gray!40!white, fill=gray!10!white, very thick},
	fairlet/.style={draw, gray!80!white, fill=gray!20!white, rounded corners, very thick}
}

\usepackage[color=blue!50]{todonotes}

\title{Improved Bounds for the Traveling Salesman Problem with Neighborhoods on Uniform Disks}
\author{}
\author{Ioana O. Bercea\thanks{Department of Computer Science, University of Maryland, College Park, USA}}
\date{}

\begin{document}
\pagenumbering{Alph}
\begin{titlepage}
\maketitle
\thispagestyle{empty}
			
	\begin{abstract}
		Given a set of $n$ disks of radius $R$ in the Euclidean plane, the Traveling Salesman Problem With Neighborhoods (TSPN) on uniform disks asks for the shortest tour that visits all of the disks. The problem is a generalization of the classical Traveling Salesman Problem(TSP) on points and has been widely studied in the literature. For the case of disjoint uniform disks of radius $R$, Dumitrescu and Mitchell~\cite{dumitrescu2003approximation} show that the optimal TSP tour on the centers of the disks is a $3.547$-approximation to the TSPN version. The core of their analysis is based on bounding the detour that the optimal TSPN tour has to make in order to visit the centers of each disk and shows that it is at most $2Rn$ in the worst case. H\"{a}me, Hyyti\"{a} and Hakula~\cite{hame-eurocg-2011} asked whether this bound is tight when $R$ is small and conjectured that it is at most $\sqrt{3}Rn$. 
		
		We further investigate this question and derive structural properties of the optimal TSPN tour to describe the cases in which the bound is smaller than $2Rn$. Specifically, we show that if the optimal TSPN tour is not a straight line, at least one of the following is guaranteed to be true: the bound is smaller than $1.999Rn$ or the $TSP$ on the centers is a $2$-approximation. The latter bound of $2$ is the best that we can get in general. Our framework is based on using the optimality of the TSPN tour to identify local structures for which the detour is large and then using their geometry to derive better lower bounds on the length of the TSPN tour. This leads to an improved approximation factor of $3.53$ for disjoint uniform disks and $6.728$ for the general case. We further show that the H\"{a}me, Hyyti\"{a} and Hakula conjecture is true for the case of three disks and discuss the method used to obtain it.
	\end{abstract}
	
\end{titlepage}
\pagenumbering{arabic}

\section{Introduction}
We study the Traveling Salesman Problem with Neighborhoods (TSPN) when each neighborhood is a disk of fixed radius $R$. The problem is a generalization of the classical Euclidean Traveling Salesman Problem (TSP), when each point to be visited is replaced with a region (interchangeably, a neighborhood) and the objective is to compute a tour of minimum length that visits at least one point from each of these regions. While it is known that Euclidean TSP admits a Polynomial Time Approximation Scheme (PTAS) due to the celebrated results of Arora~\cite{arora1998polynomial} and Mitchell~\cite{mitchell1999guillotine}, Euclidean TSPN has been shown in fact to be APX-hard~\cite{safra2006complexity,de2005tsp} even for line segments of comparable length~\cite{elbassioni2009approximation}. The geometric version of TSPN was first studied by Arkin and Hassin~\cite{arkin1994approximation} who gave constant factor approximations for a variety of cases. Since then, there has been a wide ranging study of TSPN for different types of regions. In the case of connected regions, there is a series of $O(\log n)$ approximations~\cite{mata1995approximation, elbassioni2009approximation,gudmundsson1999fast}. Better approximations are known for cases that consider various restrictions on the regions such as comparable sizes (i.e. diameter), fatness (ratio between the smallest circumscribing radius and largest inscribing radius, or how well can a disk approximate the region) and  pairwise disjointness or limited intersection~\cite{dumitrescu2003approximation,de2005tsp,elbassioni2009approximation,mitchell2010constant,chan2011qptas,dumitrescu2016traveling, bodlaender2009minimum,mitchell2007ptas,chan2018reducing}. We refer the reader to ~\cite{handbook} for a comprehensive survey of the results. 

We study the disk version which models the situation in which each customer is willing to travel a distance $R$ to meet the salesperson. This is considered an important special case of the general TSPN~\cite{dumitrescu2016traveling} and is especially relevant since it has found applicability in other areas such as path planning algorithms for coverage with a circular field of view~\cite{arkin2000approximation,galceran2013survey} and most recently for data collection in wireless sensor networks~\cite{yuan2007optimal, ma2013tour, citovsky2015exact}. Various heuristics~\cite{yuan2007optimal, carrabs2017novel,chang2015artificial,liu2013path}  and variations ~\cite{bhadauria2011robotic, ma2013tour,tekdas2012efficient} have been considered, all of which have as their basis the TSPN on uniform disks problem.

Dumitrescu and Mitchell~\cite{dumitrescu2003approximation} were the first to specifically address the case of uniform disks in 2001. They showed a PTAS for disjoint unit disks and simpler constant factor approximations for the disjoint and overlapping cases. The specific factor of $3.547$ for disjoint disks is relative to using a routine for TSP on points (i.e. the actual constant depends on the subroutine used). Later,  Dumitrescu and T\'{o}th~\cite{dumitrescu2016traveling} improved the constant factor in~\cite{dumitrescu2003approximation} for the overlapping case and extended it to unit balls in $\mathbb{R}^d$, giving a $O(7.73^d)$-approximation. When the balls are disjoint, Elbassioni et al.~\cite{elbassioni2009approximation} showed a $O(2^d/\sqrt{d})$-approximation. Most recently, Dumitrescu and T\'{o}th~\cite{dumitrescu2017constant} gave a general constant factor for (potentially overlapping) disks of arbitrary radii.  As noted by the authors, while the complexity of the disk case is well understood generally, the question of obtaining practical and better constant factor approximations remains of high  interest \cite{dumitrescu2016traveling}.

In this paper, we aim for an improved constant factor algorithm for the case of uniform radius disks and note that the algorithm proposed by Dumitrescu and Mitchell~\cite{dumitrescu2003approximation} for the disjoint case outputs an approximate TSP tour on the \textit{centers} of each disk.  The core of their analysis is a bound that compares the length of the optimal TSP on the centers of each disk($|TSP^*|$) with the length of the optimal TSPN on the disks ($|TSPN^*|$) and says that 
\begin{equation}\label{bound}
|TSP^*|\leq |TSPN^*|+2Rn, 
\end{equation} 
where $n\geq 2$ is the number of disks in the instance. In addition, the authors use a packing argument to lower bound the length of the optimal TSPN tour in terms of $R$ and $n$ and get that $ \frac{\pi}{4}Rn - \pi R  \leq |TSPN^*| $. Overall, this gives a $3.547$- approximation and in addition, the authors show that the algorithm cannot give better than a factor $2$ approximation. While other methods for choosing representative points can be employed~\cite{dumitrescu2017constant,arkin1994approximation,dumitrescu2003approximation}, this approach is appealing both in its elegance and because it does not depend on $R$. Moreover, other existing constant factors approximations for TSPN often hide large constants~\cite{de2005tsp, elbassioni2009approximation, dumitrescu2017constant} that are incurred as a consequence of using general bounds on the length of the optimal tour that do not directly exploit the structure of the regions or of the optimal TSPN tour (bounding rectangle argument and Combination Lemma in ~\cite{arkin1994approximation}). In order to improve on them, the challenge then becomes to develop bounds that exploit the difference in behavior between  a TSP tour (on points) and the TSPN tour on the regions and furthermore, avoid using general purpose techniques that add on to the overall approximation factor. 

In this context, one way to improve the approximation factor for disjoint disks is to better understand the relationship between the optimal TSP tour on the centers and the optimal TSPN on the disks. Specifically, is the $2Rn$ term in (\ref{bound}) tight or can it be improved by using specific structural properties of the optimal TSPN? A similar question was asked in 2011 by H\"{a}me, Hyyti\"{a} and Hakula~\cite{hame-eurocg-2011} for the case when $R$ is very small (and hence, $TSP^*$ and $TSPN^*$ respect the same order and the disks are pairwise disjoint). They conjectured that the true detour term should be $\sqrt{3}nR$ and constructed arbitrarily large instances of disjoint disks that converge to this case. We refer to this as the \textbf{H\"{a}me, Hyyti\"{a} and Hakula conjecture.} Subsequent experiments by M\"{u}ller~\cite{muller2011finding}, however, suggest that this might be true only for tours up to five disks and higher otherwise. No further progress has been made towards the conjecture since then.

\paragraph*{Contributions.}
We make the first progress on the conjecture and develop a twofold method that either improves the bound in (\ref{bound}) or 
shows that the TSP on the centers is a good approximation for the TSPN on the disks. Formally, we get that

\begin{theorem}\label{thm:mama}
	For any $n \geq 4$ disjoint disks of radius $R$ at least one of the following is true:
	\begin{itemize}
		\item $TSPN^*$ is supported by a straight line,
		\item $|TSP^*| \leq |TSPN^*| + 1.999Rn$,
		\item $|TSP^*| \leq 2\cdot |TSPN^*|$.
	\end{itemize} 
	Our framework also gives an overall $3.53$-approximation for the case of uniform disjoint disks and a $6.728$-approximation for the overlapping case.
\end{theorem}

While the improvement in the overall approximation factor is small, our framework strives to provide new insight into the problem that can be explored further. Specifically, the $2$-approximation (optimal with respect to the method of computing a TSP on the centers~\cite{dumitrescu2003approximation}) comes from the case in which the TSPN tour takes a lot of sharp turns. Furthermore, it is based on a lower bound that does not rely on packing arguments. To the best of our knowledge, this is the first such bound specifically for TSPN out of all arguments for general fat regions~\cite{mitchell2010constant}. As such, it might be of independent interest and it could, for example, lead to improved approximation factors for balls in $\mathbb{R}^d$ that do not depend exponentially on the dimension. Moreover, the fatness of the disks is used in showing that short sharp turns lead to a disk being visited multiple times and can conceivably be used to show similar properties for other fat regions.

We start by fixing an order $\sigma$ and comparing the TSPN tour that visits the disks in that order to the TSP tour that visits their respective centers in the same order. The $2Rn$ term in (\ref{bound}) comes from considering the points at which the TSPN touches the boundary of each disk and charging each such vertex with a $2R$ detour for going to its respective center and coming back. In this view, the $2Rn$ term cannot be improved since the charge on each vertex will always be $2R$. Instead, we reinterpret the bound as charging the \textit{edges} of the TSPN tour instead of its vertices and notice that the charge for each edge can now be anywhere between $-2R$ and $2R$, depending on how close the tour is to (locally) visiting pairs of disks optimally. In this context, we define a ``bad'' edge to be one that incurs a large charge (i.e. $>(2-\epsilon)R$ for some $\epsilon>0$). We show that such bad edges lead to the TSPN tour exhibiting sharp turns (i.e. with small interior angle). When the edges of the sharp turn are long, we use that to derive a better lower bound on the overall TSPN tour length. On the other hand, when one of them is short, we show that the tour must then visit a disk twice (i.e. visit it once, then touch another disk and return back to it). The crux of the argument is in understanding how these short sharp turns that visit a disk multiple times influence the global detour term. 

When a tour visits a disk more than once, two scenarios follow naturally from the classical TSP case of just visiting points: either the order $\sigma$ is not optimal or the tour must follow a straight line. Surprisingly, we show that a third alternative scenario is also possible, whose local structure we call a $\beta$-triad. The main technical contribution of the paper is in describing structural properties of such $\beta$-triads and showing that they actually have a low \textit{average} detour. Specifically, we construct an additional order $\sigma'$ and use an averaging argument to show that $\beta$-triads have low detour when compared to the TSP tours that visit the centers in the order $\sigma$ and $\sigma'$. This then allows us to conclude that they have a low detour with respect to the optimal TSP on the centers.


Along the way, we also show that the H\"{a}me, Hyyti\"{a} and Hakula conjecture is true for $n=3$ and use it to bound the average detour of $\beta$-triads.  We include a discussion of the method used to derive it, involving Fermat-Weber points, which might be useful for the case of $n\geq 4$. We also discuss how our approach can be used within the framework of Dumitrescu and T\'{o}th~\cite{dumitrescu2016traveling} to yield improved approximation factors for the overlapping disks case.

\paragraph*{Preliminaries.}
We consider $n \geq 3$ disjoint disks of radius $R$ in the Euclidean plane.  We denote an optimal TSP tour on the centers of the disks as $TSP^*$. Similarly, $TSPN^*$ will denote an optimal TSPN tour on the disks. Our results will be with respect to a fixed TSPN tour (which we call simply $TSPN$) described by a sequence of ordered points $P_1, P_2, \ldots, P_n$  on the boundary of the disks such that the tour is a polygonal cycle with edges $(P_i, P_{i+1})$. Furthermore, we have that for each of the input disks, there exists some $i \in [1,n] $ such that point $P_i$ is on the boundary of the disk. 

Notice that the points $P_i$ induce a natural order $\sigma$ on the disks with centers $O_1, O_2, \ldots O_n$, i.e. $\sigma$ corresponds to the identity permutation on $P_1, \ldots, P_n$. For the majority of our theorems, we will assume that $TSP$ always refers to a tour \textit{on the centers} and in the order $\sigma$ on the disks. When we need to make a difference, we will further use $TSP(\sigma')$ to be the tour which visits the centers in the order given by the permutation $\sigma'$. Given two such permutations $\sigma$ and $\sigma'$, we say that $\sigma \cap \sigma'$ refers to the maximal set of points on which $\sigma$ and $\sigma'$ agree. In this context, $TSPN(\sigma \cap \sigma')$ refers to the collection of paths we get from visiting the points $P_i$ according to $\sigma \cap \sigma'$. Similarly, $TSP(\sigma \cap \sigma')$ corresponds to the collection of paths that we get from visiting the points $O_i$ according to $\sigma \cap \sigma'$.

Finally, we denote the length of a tour $\mathcal{T}$ as $|\mathcal{T}|$. When $\mathcal{T}$ is a collection of paths, we have that $|\mathcal{T}|$ represents the total length of each of the paths. When $A$ and $B$ are points, we have that $|AB|$ denotes the length of the segment $AB$.  We therefore have that $|TSPN| = \sum_{i=1}^{n} |P_iP_{i+1}|$ and $|TSP| = \sum_{i=1}^{n} |O_iO_{i+1}|$, where $P_{n+1}=P_1$ and $O_{n+1}=O_1$.

	\section{$\beta$-triads and a Structural Theorem}
	\label{43}

Before we formally define what a ``bad'' edge is, we will describe how to interpret the $2Rn$ \textit{detour bound} from ~\cite{dumitrescu2003approximation} as charging edges instead of vertices. We fix an order $\sigma$ and consider the points $P_i$ and $O_i$ as previously defined. The argument in ~\cite{dumitrescu2003approximation} then says that we must have:
\begin{center}
 $\sum_i |O_iO_{i+1}| \leq \sum_i |P_iP_{i+1}| + 2Rn$.
\end{center}

In this context, the term $ \sum_i P_iP_{i+1} + 2Rn$ is the length of a tour that follows the TSPN tour and additionally, at each point $P_i$, takes a detour of $2R$ to visit the center $O_i$ and come back. Choosing $\sigma$ to be the optimal order in which $TSPN^*$ visits the disks gives us (~\ref{bound}). In this view, the detour term $2Rn$ is obtained by charging $2R$ to each point $P_i$ of the TSPN tour. Instead, we can also think of it as coming from charging each \textit{edge} $P_iP_{i+1}$ of the tour with a \textit{local detour} of $2R$ in the following sense:
\begin{center}
$|O_iO_{i+1}| \leq |P_iP_{i+1}| + 2R$.
\end{center}

This new perspective is quite natural since it captures the observation that the shortest edge which visits the disks centered at $O_i$ and $O_{i+1}$ has length exactly $|O_iO_{i+1}| - 2R$ and hence the $TSPN$ tour has to pay at least that for each pair of consecutive disks it visits. In this sense, we decompose the global detour term of $2Rn$ into $n$ local detour terms $|O_iO_{i+1}| - |P_iP_{i+1}|$ that essentially quantify how efficient the TSPN on the disks is locally. 

In this context, saying a TSPN edge has a high local detour is equivalent to saying that it is close to being locally optimal or shortest possible: when the edge is exactly of length $|O_iO_{i+1}| -2R$, its local detour is $2R$ (the maximum). If, on the other hand, we know that the edge is bounded away from $|O_iO_{i+1}| -2R$, i.e. $|P_iP_{i+1}| > |O_iO_{i+1}| -2R + \epsilon R$, for some $\epsilon>0$, this translates into a local detour of at most $(2-\epsilon)R$. Intuitively, such an edge is ``good'' for us because it allows us to lower the overall detour term. In contrast, ``bad'' $P_iP_{i+1}$ edges are the ones for which the local detour term is large and consequently, their length is closer to $|O_iO_{i+1}| -2R$. Our technique is motivated by trying to describe the behavior of such bad edges.

Formally, we consider a fixed angle parameter $\beta \in [0,\pi/12]$ that we instantiate later when we derive the overall bounds. We define the function:
\begin{center}
$f(O_1O_2,\beta) =\sqrt{|O_1O_2|^2 + R^2 - 2R|O_1O_2|\cos \beta}$,
\end{center}
which is $|O_1O_2| - R$ when $\beta =0$ and $|O_1O_2| + R$ when $\beta = \pi$. Intuitively, the quantity $f(O_1O_2,\beta) - R$ will control how close we are to $|O_1O_2| - 2R$.  We then  say that the edge $P_1P_2$ is \textbf{bad} if $|P_1P_2| \leq  f(O_1O_2,\beta)-R$ and \textbf{good} otherwise (we abstract away the dependency on $\beta$ for simplicity). Bad edges are close to $|O_1O_2| - 2R$ and will incur a large local detour. In contrast, using straightforward algebra, one can show that a good edge $P_1P_2$ is guaranteed to have a small detour:  $|O_1O_2| \leq |P_1P_2| + (1+\cos \beta)R$.

\subsection{Consecutive Bad Edges}
\begin{figure}
\centering
\renewcommand{\baselinestretch}{1}
	\small\normalsize
  \includegraphics[height = 1.75in]{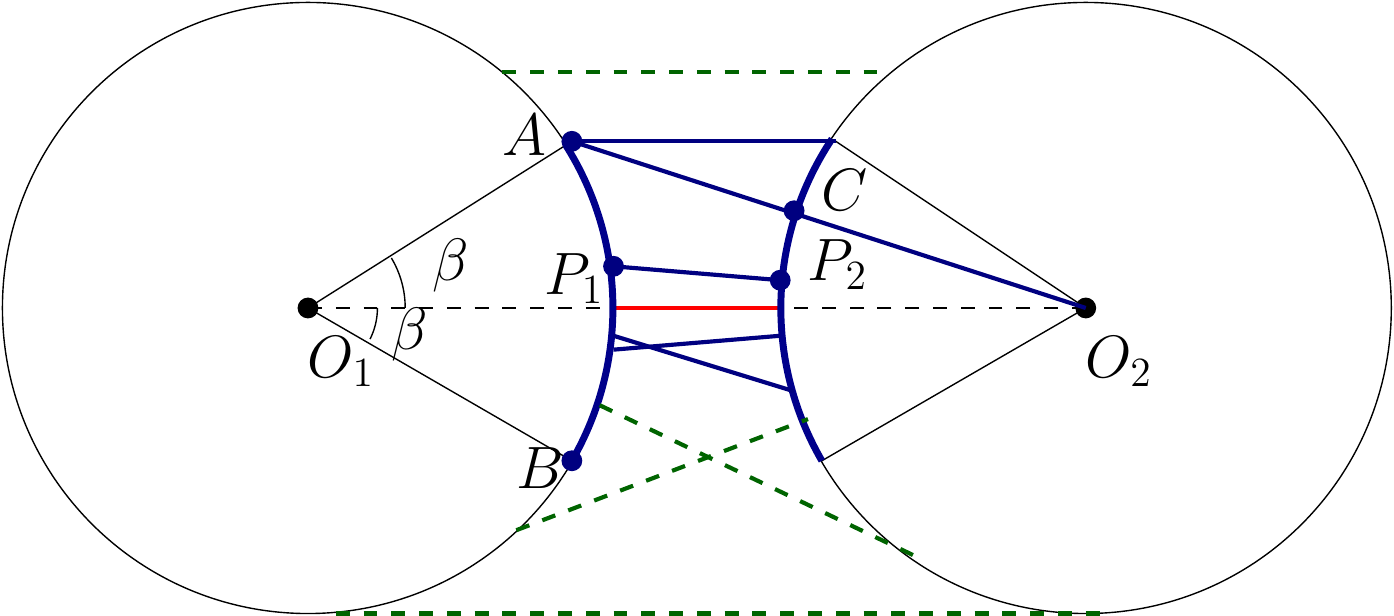}
 \caption[Description of a bad edge.]{Bad edges are guaranteed to have both endpoints in the blue arcs: if $|P_1P_2| \leq |AC|$, then $P_1P_2$ is a bad edge. In contrast, the dashed edges are guaranteed to be good edges.}\label{fig:beta}
\end{figure}
The idea behind defining bad edges in terms of $f(O_1O_2,\beta)-R$ is that it allows us to restrict the location of $P_1$ and $P_2$ on the boundary of their respective disks as seen in Figure ~\ref{fig:beta}. Specifically, there are exactly two points $A$ and $B$ on the boundary of the first disk with the property that the shortest distance from $A$ or $B$ to the boundary of the second disk is exactly  $f(O_1O_2,\beta)-R$. Not coincidentally, they form an angle of $\beta$ with $O_1O_2: \angle AO_1O_2 = \angle BO_1O_2= \beta$. In general, $P_1$ (and in a similar fashion $P_2$) is guaranteed to lie in the short arc between $A$ and $B$ whenever $P_1P_2$ is upper bounded by  $f(O_1O_2,\beta)-R$:

\begin{lemma} \label{lemma:angle} If $P_1P_2$ is bad, then the angles $\angle O_1O_2P_2$ and $\angle O_2O_1P_1$ are $\leq \beta$.
\end{lemma} 
\begin{proof}
 Let $\gamma = \angle O_1O_2P_2 \leq \pi$ and notice that $O_1P_2= f(O_1O_2,\gamma)$. Consider the point $Q$ where $O_1P_2$ intersects the first disk and note that the shortest distance from $P_2$ to the first disk is exactly $P_2Q = f(O_1O_2,\gamma)-R$. We therefore get that $P_1P_2 \geq P_2Q$.
 Now notice that, if $\gamma > \beta$, then $f(O_1O_2,\gamma) > f(O_1O_2, \beta)$ and so $P_2Q > f(O_1O_2,\beta)-R$, which would lead to a contradiction. The same argument can be applied for $P_2$ and we get our conclusion.
\end{proof}

When a second bad edge $P_2P_3$ is considered, we can conclude that the angle $O_1O_2O_3$ has to be at most $2\beta$ and hence the TSP on the centers must make a sharp turn after it visits $O_2$. Specifically, let $O_3$ be the center of the disk visited next at $P_3$ and assume that the edge $P_2P_3$ is also bad. Notice that the angle $\angle O_1O_2O_3$ formed by the TSP is either $\angle O_1O_2P_2 + \angle P_2O_2O_3$ or $|\angle O_1O_2P_2 - \angle P_2O_2O_3|$. Regardless, we have that $\angle O_1O_2O_3 \leq \angle O_1O_2P_2 + \angle P_2O_2O_3$ and get the following corollary: 

\begin{cor}\label{cor:angle}
If both $P_1P_2$ and $P_2P_3$ are bad edges, then the angle $\angle O_1O_2O_3$ is $\leq 2\beta$. 
\end{cor}

If that happens and the disks are close to each other, we have that one of the edges of the TSP must actually intersect a disk twice. Specifically, if $|O_1O_2| \leq R/\sin(2\beta)$, then the support line for $O_2O_3$ must pass through the disk centered at $O_1$. In general, it is not true that if $O_2O_3$ intersects the first disk, we immediately get that the associated TSPN edge $P_2P_3$ must also intersect it. In our case, however, we have that the slope of $P_2P_3$ is very close to the one of $O_2O_3$ due to the fact that it is a bad edge. We use this information to show that if $O_2O_3$ does not intersect the first disk, then $P_2P_3$ cannot be a bad edge.

\begin{theorem}\label{thm:intersect}
  If $P_1P_2$ and $P_2P_3$ are bad edges and $O_1O_2 \leq R/\sin(2\beta)$, then the segment $P_2P_3$ intersects the disk centered at $O_1$.
\end{theorem}

\begin{proof}
We consider the case in which $\angle O_1O_2O_3 = \angle O_1O_2P_2 + \angle P_2O_2O_3$ and note that all the other cases are similar. We denote the two lines originating at $O_2$ that are tangent to the first circle as $\ell_1$ and $\ell_2$ such that the line $O_2O_3$ is in between $\ell_1$ and $O_2O_1$. Note that this is possible because the angle that $\ell_1$ forms with $O_2O_1$ is at least $2\beta$ (since $O_1O_2 \leq R/\sin(2\beta)$) but the angle that $O_2O_3$ forms with $O_2O_1$ is at most $2\beta$ (Corollary~\ref{cor:angle}).

Our strategy will be to first show that the segment $P_2P_3$ is contained in the wedge defined by $\ell_1$ and $\ell_2$ (Figure ~\ref{fig:intersect}). Notice that, since the wedge defines a convex space, it is enough to show that $P_2$ and $P_3$ are contained in it. 

We first show that the point $P_2$ has to be in the wedge.  Let $S_1$ and $S_2$ be the points in which the segment $O_1O_2$ intersects the first and second disk. Similarly, let $T_2$ and $T_3$ be the points in which $O_2O_3$ intersects the second and third disk. We then have that $P_2$ is between $T_2$ and $S_2$.

Now we only need to show that $P_3$ is in between $\ell_1$ and $\ell_2$. We will do that by arguing that any choice of $P_3$ outside of the wedge will contradict the fact that $P_2P_3$ is a bad edge.

\begin{figure}[h]
\centering
\renewcommand{\baselinestretch}{1}
	\small\normalsize
  \includegraphics[height = 2in]{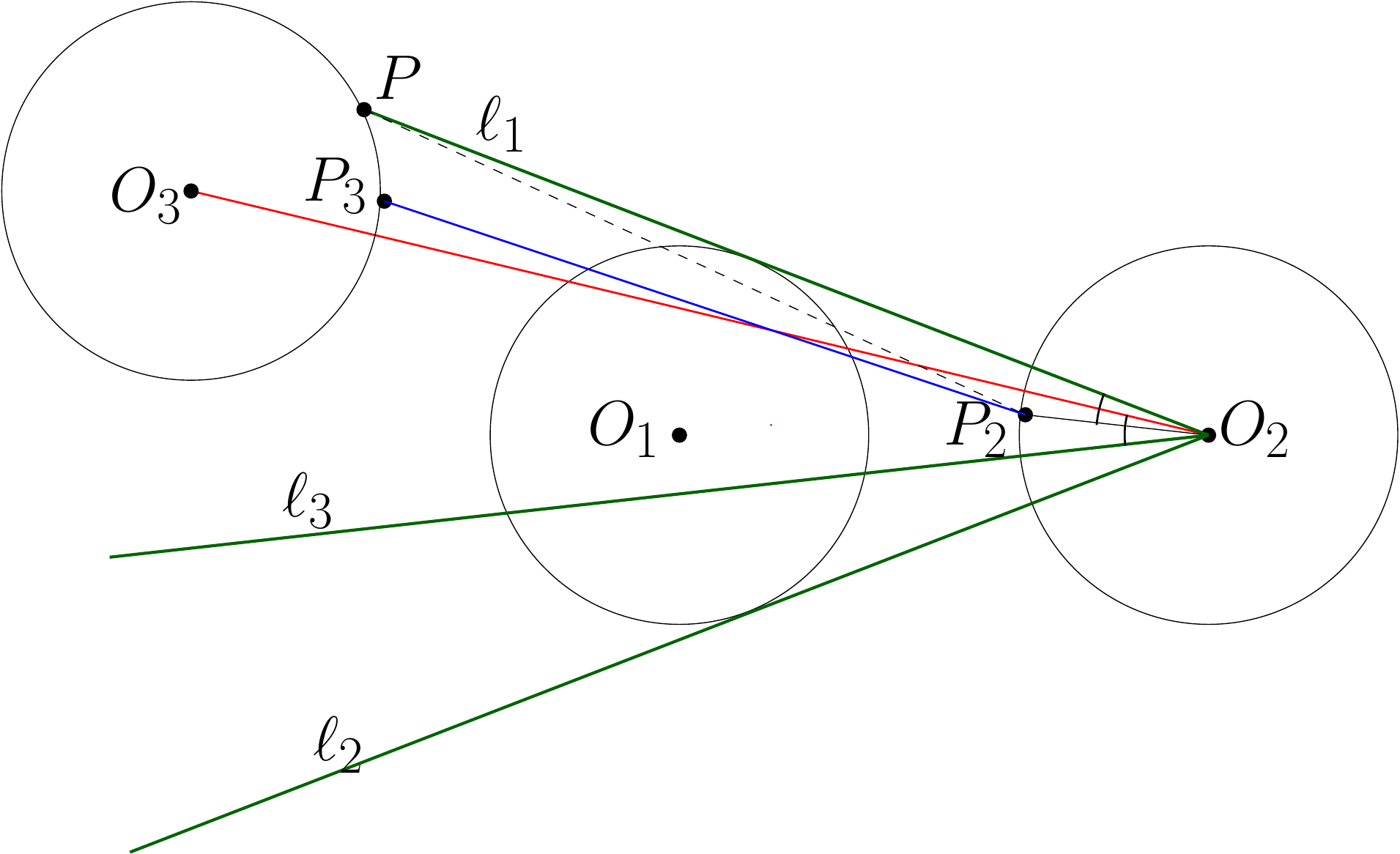}
 \caption[Argument for the proof of Lemma 4.2]{When $O_2O_3$ crosses the disk centered at $O_1$, we must also have that the segment $P_2P_3$ also crosses it. We show this by arguing that $P_2P_3$ is contained between the two lines $\ell_1$ and $\ell_3$ and that $P_2$ and $P_3$ are on separate sides of the first disk.}\label{fig:intersect}
\end{figure}

 Let $\alpha_1 = \angle O_1O_2P_2 $ and $\alpha_2 =  \angle P_2O_2O_3$, and so $\alpha_1, \alpha_2 \leq \beta$ (Lemma~\ref{lemma:angle}). First notice that if neither $\ell_1$ nor $\ell_2$ intersect the third disk, then we are done because we have that the entire boundary is contained in the convex space (since $O_3$ is already in between $\ell_1$ and $\ell_2$). Assume then that $\ell_1$ intersects the third disk at a point $P$ above the line $O_2O_3$, since $\angle O_3O_2O_1 \leq 2\beta \leq \angle PO_2O_1$. Moreover, since $\angle PO_2O_1 \geq 2\beta$, we have that $\angle PO_2P_2 \geq 2\beta - \alpha_1 \geq \beta$ and so $PP_2 \geq f(PO_2,\beta)$ (because $|P_2O_2| = R$). Since $|PO_2|\geq |T_3O_2|$ and $|T_3O_2| \geq R$, this implies that $|PP_2| \geq f(T_3O_2,\beta)$. Using the fact that $f(x,\beta) \geq x-R\cos \beta$ for any $x$ and $\beta \neq 0$, one can verify that:
 \begin{align*}
 f(T_3O_2,\beta) &= f(O_2O_3 -R, \beta) \\
                  &= \sqrt{(|O_2O_3|-R)^2 + R^2 - 2R(|O_2O_3|-R)\cos \beta}\\
                  &> \sqrt{|O_2O_3|^2 + R^2 - 2R|O_2O_3|\cos \beta}-R\\
                  &> f(O_2O_3,\beta)-R.
 \end{align*}
This means that $P$ cannot be a possible position for $P_3$ because then $P_2P_3$ would be too big. Moreover, any point $Q$ "above" $P$ (i.e. such that $\angle O_3O_2Q > \angle O_3O_2P$) would also not work as a possible position for the same reason. In other words, $P_3$ has to be underneath the line $PO_2 = \ell_1$. 

In order to prove that $P_3$ is also above the line $\ell_2$, we will consider an additional line $\ell_3$ originating at $O_2$ that makes an angle of $\beta$ with $O_2P_2$ and is underneath it. This new line makes an angle of $\beta+\alpha_2$ with $O_2O_3$ and since  $\ell_2$ makes an angle of $\geq 2\beta+\alpha_1+ \alpha_2$ with $O_2O_3$, we get that $\ell_3$ is in between $O_2O_3$ and $\ell_2$. In other words, if we show that $P_3$ is above $\ell_3$, then we also get that $P_3$ is above $\ell_2$. If $\ell_3$ does not intersect the third disk, then we are done as before, so assume that it intersects it at a point $Q$ on the boundary. Similarly as before, we have that $P_2Q = f(O_2Q,\beta) \geq f(O_2T_3,\beta) > f(O_2O_3,\beta)-R$. This in turn implies that $P_3$ has to be above $Q$, otherwise $P_2P_3$ would be too big. Therefore $P_3$ must be above the line $\ell_3$.

At this point, we have that the segment $P_2P_3$ is contained in the wedge defined by $\ell_1$ and $\ell_2$. We know that the first disk is tangent on both sides to $\ell_1$ and $\ell_2$ but this does not directly imply that $P_2P_3$ must actually intersect it. In order to have that, we must also ensure that $P_2$ and $P_3$ lie on different sides of the first disk. We argue this by showing that $O_3$ itself must be on the other side of the first disk as $O_2$. Since the disks do not intersect, this implies that $P_3$ is on a different side from $P_2$. In order to show this, notice that we can assume, without loss of generality, that $O_1O_2 \leq O_2O_3$. Let $T$ be the point on $O_2O_3$ such that $O_1T \perp O_2O_3$. Since $\angle O_1O_2O_3 \leq 2\beta$ and  $O_1O_2 \leq R/\sin(2\beta)$, this means that $T$ is contained in the first disk. Suppose that $O_3$ is on the segment $O_2T$ (effectively in between $O_1$ and $O_2$). Then $O_2O3 < O_2T$ but, since $O_2T = O_1O_2 \cos(\angle O_1O_2T) \leq O_1O_2$, this would lead to a contradiction. We therefore get that $O_2$ and $O_3$ are on different sides of the first disk and that the same is true for $P_2$ and $P_3$. This shows that the segment $P_2P_3$ must intersect the first disk.
\end{proof}

\subsection{Introducing $\beta$-triads}

The fact that the disk centered at $O_1$ is crossed by both $P_1P_2$ and $P_2P_3$ suggests that the TSPN might not be optimal because it could be shortcut. Our structural theorem identifies when that is the case and isolates the remainder as having a specialized local structure which we call a $\beta$-triad. Formally, we say that a specific TSPN subpath  $P_n - P_1 - P_2 - P_3$ is a \textbf{$\beta$-triad} if it satisfies all of the following properties (Figure~\ref{fig:triadpic}):
\begin{itemize}
\item $P_1P_2$ and $P_2P_3$ are bad edges and $O_1O_2 \leq R/\sin(2\beta)$,
\item $P_1, P_2, P_3$ are not collinear but $P_n,P_1,P_2$ are collinear with $P_1$ between $P_n$ and $P_2$. 
\end{itemize}

\begin{figure}[h]
\centering
\renewcommand{\baselinestretch}{1}
	\small\normalsize
  \includegraphics[height = 2in]{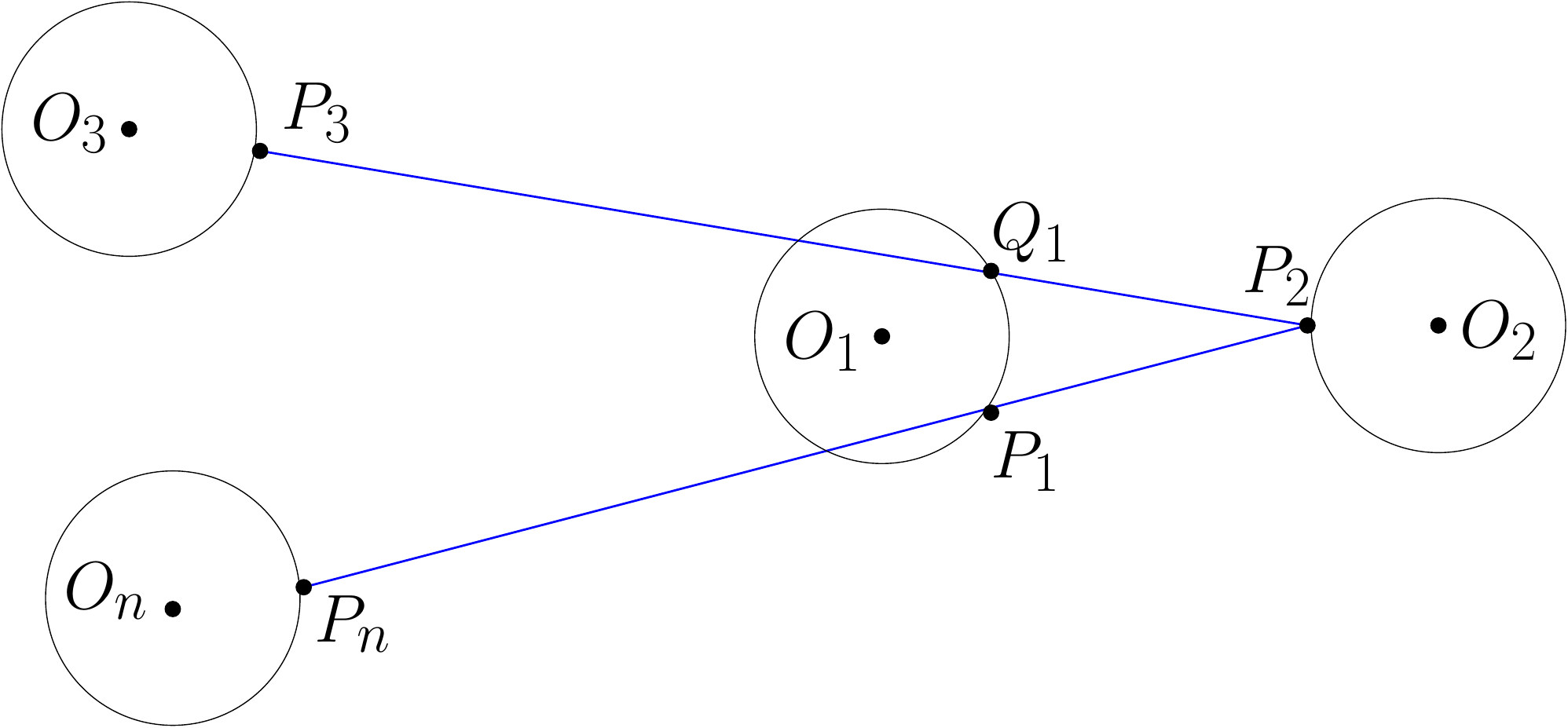}
 \caption[Definition of a $\beta$-triad.]{The path $P_n - P_1 - P_2 - P_3$ forms a $\beta$-triad.}\label{fig:triadpic}
\end{figure}

We state the structural theorem here and refer the reader to Appendix~\ref{app:main} for a complete argument. The case in which the TSPN tour follows a straight line that stabs all the disks is discussed separately in Section BLAH BLAH BLAH and is of separate interest.

\begin{theorem}\label{thm:main}
 For $n\geq 4$, if $P_1P_2$ and $P_2P_3$ are bad edges  and $O_1O_2 \leq R/\sin(2\beta)$ then at least one of the following is true:
 \begin{itemize}
 \item the TSPN tour is not optimal,
 \item the TSPN tour is supported by a straight line or
 \item the path $P_n - P_1 - P_2 - P_3$ forms a $\beta$-triad.
 \end{itemize}   
\end{theorem}

\begin{proof}
 We distinguish between the case in which $P_2P_3$ intersects the first disk at $P_1$ and otherwise. In the first case, we will show that either the TSPN is not optimal or all the disks are stabbed by it. The second case is more involved and reduces to describing what the local structure of the TSPN must be such that it does not necessarily fall in the previous two cases.

\textbf{Case $1$: $P_1,P_2,P_3$ are collinear.} Then consider the point $P_n$ that connects to $P_1$. The cost that the TSPN pays for visiting the four disks is $|P_nP_1|+|P_1P_2|+|P_2P_3|$ but by triangle inequality, we know that $|P_nP_2| \leq |P_nP_1| + |P_1P_2|$, so the TSPN would visit $P_2$ directly and pass through $P_1$ on its way to $P_3$. If the inequality is strict, then this directly implies that the TSPN is not optimal. When we have equality, however, this implies that $P_n,P_1$ and $P_2$ are now also collinear and furthermore, that $P_1$ lies between $P_2$ and $P_n$. In other words, we have that on the line from $P_2$ to $P_n$, we have both $P_3$ and $P_n$ to the left of $P_1$. Now look at how point $P_4$ connects to $P_3$ and notice that the portion of TSPN for the five disks is now $|P_4P_3| + |P_3P_2|+|P_1P_2|+|P_1P_n|$ and again, we can ask the question of why wouldn't the TSPN go straight to $P_2$ instead and visit $P_3$ along the line $P_2P_3$. Specifically, we have $|P_4P_2| \leq |P_4P_3| + |P_3P_2|$ with the TSPN not being optimal whenever this inequality is strict. We therefore consider the case in which $|P_4P_2|= |P_4P_3| + |P_3P_2|$ and get that now $P_4$ has to also be collinear with the other points and furthermore, $P_3$ has to be between $P_4$ and $P_2$. Continuing this process, we get that all the TSPN points would have to be collinear and in the order $P_2,P_1,P_3,P_4,\ldots, P_{n-1}$ with $P_n$ potentially being anywhere past $P_1$. In this case, we have that the TSPN is a supported by a straight line that stabs all of the disks.

\textbf{Case $2$: $P_1,P_2,P_3$ are not collinear.} Let the line $P_1P_2$ intersect the first disk for the first time at $Q_1$. By the argument from before, we know that if $|P_nP_2| < |P_nP_1| + |P_1P_2|$, then the TSPN cannot be optimal since another tour could go from $P_n$ straight to visiting $P_2$ and then visit $P_1$ on the way to $P_3$, at a lesser cost. When $P_n,P_1$ and $P_2$ are collinear, in that order, we say that $P_n - P_1 - P_2 - P_3$ form a $\beta$-triad.

\end{proof}

\subsection{Properties of $\beta$-triads}
Theorem ~\ref{thm:main} says that if $|TSPN^*|$ is not a straight line, then the triad has a local detour of at most $3\sqrt{3}R$. Lemma ~\ref{lem:betatriad} further states that all the bad triads are also edge disjoint.  In order to prove that, we go back to the proof of Theorem~\ref{thm:main}. Note that we distinguished between the case in which $P_1,P_2$ and $P_3$ are collinear (Case $1$) and when they are not(Case $2$). The first case leads to the TSPN being a straight line, which is ruled out by our assumptions. In the second case, the optimality of $TSN^*$ implies that $P_n, P_1$ and $P_2$ are also collinear, with $P_1$ between $P_2$ and $P_n$.

\begin{lemma}\label{lem:betatriad}
All the $\beta$-triads in a given TSPN tour are edge disjoint.
\end{lemma}

\begin{proof}
Assume there is another bad triad that shares edges with $P_n - P_1 - P_2 - P_3$. We distinguish four cases, based on the type of edges they have in common.

\textbf{Case $1$: $P_{n-1} - P_n - P_1 - P_2$ is a bad triad.} 
This case cannot happen since $P_n,P_1,P_2$ are collinear.

\textbf{Case $2$: $P_{n-2} - P_{n-1} - P_n - P_1$ is a bad triad.} Then, by definition, we must have that $P_nP_1$ is also a bad edge. We will show, however, than this cannot be. For this, we will use an additional lemma:

\begin{figure}[h]
\centering
\renewcommand{\baselinestretch}{1}
	\small\normalsize
  \includegraphics[height = 2in]{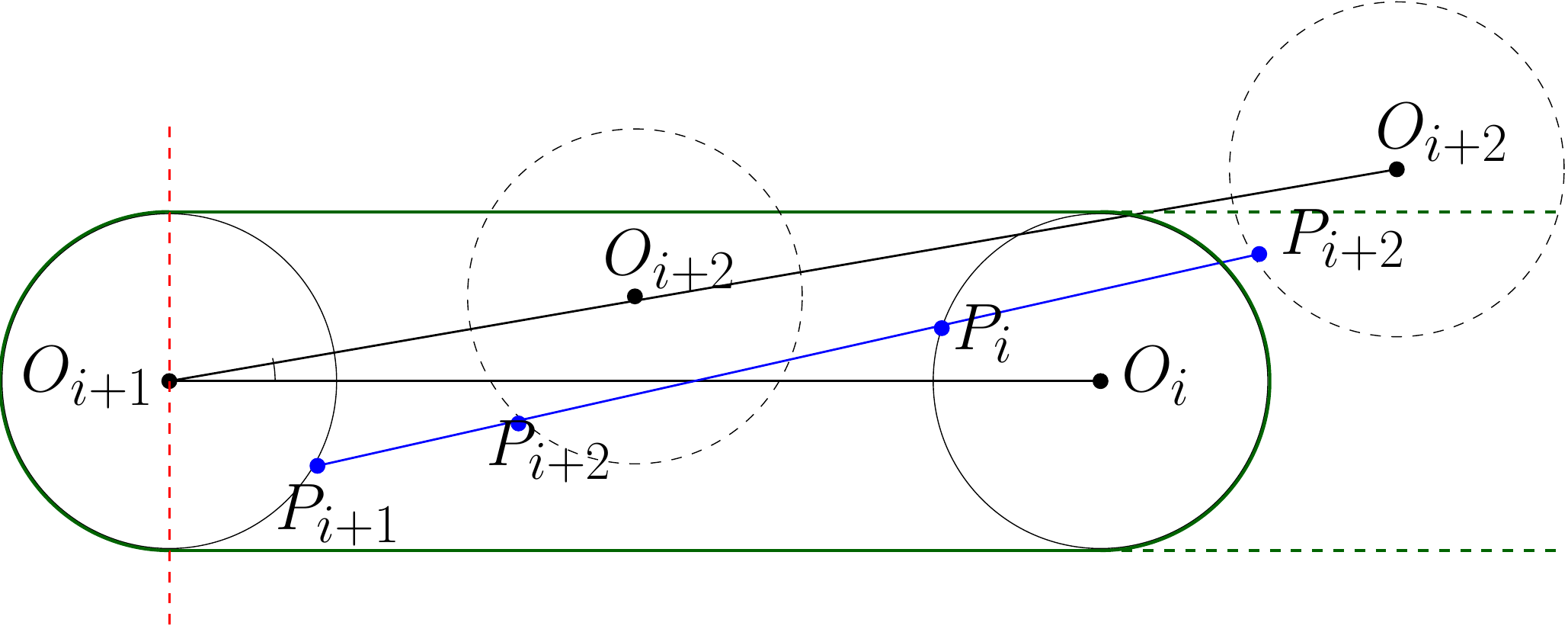}
 \caption[Argument for the proof of Lemma 4.2.]{A potential TSPN path is drawn in blue. Because the angle $\angle O_iO_{i+1}O_{i+2}  \leq \pi/2$, we have that $O_i$ and $O_{i+1}$ are on the same side of the hyperplane described by the red line. That, in turns, give two options for the disk centered at $O_{i+2}$ to intersect the (extended) convex hulls of the other two diskss, drawn in green. In each case, the points are visited in the wrong order.}\label{fig:convexhull}
\end{figure}

\begin{lemma}\label{lem:disjoint} If $P_iP_{i+2}$ is a straight line that passes through point $P_{i+1}$ such that $P_{i+1}$ is between $P_i$ and $P_{i+2}$, then it cannot be that both $P_iP_{i+1}$ and $P_{i+2}P_{i+1}$ are bad edges. 
\end{lemma}
\begin{proof}
Assume that both $P_iP_{i+1}$ and $P_{i+2}P_{i+1}$ are bad edges. Then Corollary ~\ref{cor:angle} implies that the angle $\angle O_{i}O_{i+1}O_{i+2} \leq 2\beta$. Now consider the convex hull of the two disks centered at $O_i$ and $O_{i+1}$ (Figure ~\ref{fig:convexhull}). If the disk centered at $O_{i+2}$ intersects the convex hull, then $P_{i+2}$ must be contained in that convex hull, otherwise the line $P_i - P_{i+1} - P_{i+2}$ would not exist. But in that case, the points would be visited out of order. Specifically, $P_{i+2}$ would be between $P_i$ and $P_{i+1}$. 

Now assume that the disk centered at $O_{i+2}$ does not intersect the convex hull. Since the angle $\angle O_{i}O_{i+1}O_{i+2} \leq 2\beta \leq \pi/6$, this implies that $O_{i+2}$ is in the same halfspace as $O_{i}$ with respect to the line perpendicular to $O_iO_{i+1}$ passing through $O_{i+1}$. We extend the convex hull infinitely in that halfspace by allowing the tangent lines to be infinite on that side. By the same argument as before, we know that the disk centered at $O_{i+2}$ must intersect this extended region. But then we would get again that the points are out of order: $P_{i}$ would be between $P_{i+1}$ and $P_{i+2}$.  
\end{proof}

When $P_i = P_n, P_{i+1} = P_1$ and $P_{i+2} = P_2$, Lemma ~\ref{lem:disjoint} tells us that it cannot be that $P_1$ is between $P_n$ and $P_2$ and both edges $P_1P_n$ and $P_1P_2$ are bad. Therefore we are done with this case.

\textbf{Case $3$: $P_1 - P_2 - P_3 - P_4$ is a bad triad.} This case is similar to Case $1$ and cannot happen, since $P_1, P_2$ and $P_3$ cannot be collinear.
\textbf{Case $4$: $P_2 - P_3 - P_4 - P_5$ is a bad triad.} This case is similar to Case $2$ because we have that $P_2$, $P_3$ and $P_4$ are collinear with $P_3$ between $P_2$ and $P_4$ and both $P_2P_3$ and $P_3P_4$ being bad edges. 

\end{proof}

\subsection{Structural Theorem} 
The case of $\beta$-triads is interesting because it arises naturally as a consequence of dealing with regions instead of points. Given any optimal tour that exhibits internal angles $\leq \pi/6$, we can always add an extra disk at each sharp turn that will maintain optimality, pairwise disjointness and be intersected twice by this tour, giving rise to a $\beta$-triad. It is therefore important that we understand their behaviour. 

Because of the fact that $P_1P_2$ and $P_2P_3$ are bad edges, the $\beta$-triad is likely to have a high detour with respect to $TSP(\sigma)$. Nevertheless, we show that there exists an alternate ordering $\sigma'$ such that the average detour of the three edges in the $\beta$-triad with respect to $(|TSP(\sigma)| + |TSP(\sigma')|)/2$ is $3\sqrt{3}R$. The order $\sigma'$ takes advantage of the fact that the disk centered at $O_1$ is crossed twice and inverts the order in which it is visited without changing the cost of the underlying TSPN tour. The $3\sqrt{3}R$ bound comes from proving the H\"{a}me, Hyyti\"{a} and Hakula conjecture for $n=3$ (Section ~\ref{fermat}). In order to be able to construct $\sigma'$ consistently across multiple $\beta$-triads, we also show that $\beta$-triads are isolated events and specifically that they are edge-disjoint (Lemma ~\ref{lem:betatriad}).

\begin{theorem}\label{thm:triadbound} If the TSPN in the order $\sigma$ has $k$ $\beta$-triads that together cover a set of edges of total length $L_T$, then we can construct another order $\sigma'$ that agrees with $\sigma$ on everything except the order inside the $\beta$-triads such that:
\begin{center}
 $\displaystyle{ \frac{|TSP(\sigma)| + |TSP(\sigma')|}{2} \leq |TSP(\sigma \cap \sigma')| + L_T+ 3\sqrt{3}Rk}$.
\end{center}
\end{theorem}

\begin{proof}
We discuss the case for $k=1$ and show how to modify the argument for $k>1$. Suppose that $P_n - P_1 - P_2 - P_3$ form a $\beta$-triad. We know that $TSP(\sigma)$ visits the centers of each disk in the order $O_1, O_2, O_3, \ldots O_n$. We consider an additional order $\sigma'$ such that $TSP(\sigma')$ visits the centers in the order $O_2, O_1,O_3, \ldots O_n$. Notice that both $\sigma$ ad $\sigma'$ agree on the order $O_3,O_4, \ldots, O_{n}$ and that they differ in the fact that $\sigma$ visits $O_1$ after $O_n$ and before $O_2$ and $\sigma'$ visits $O_1$ after $O_2$ and before $O_3$. Therefore,  for $T' =TSP(\sigma \cap \sigma')$, we have that  $|T'| = |O_3O_4|+\ldots+|O_{n-1}O_n|$,  $|TSP(\sigma)| = |T'| +  |O_nO_1| + |O_1O_2| + |O_2O_3|$ and $|TSP(\sigma')| = |T'| +  |O_nO_2| + |O_2O_1| + |O_1O_3|$.

On the other hand, the length of the TSPN with respect to the orders $\sigma$ and $\sigma'$ stays the same. The local cost of visiting $P_n - P_1 - P_2 - P_3$ is $L_T = |P_nP_1| + |P_1P_2| + |P_2P_3| = |P_nP_2| + |P_2P_3|$, since $P_n, P_1$ and $P_2$ are collinear and $P_1$ is between $P_n$ and $P_2$. We also know that $P_2P_3$ intersects the disk centered at $O_1$ at some point $Q_1$ that is different from $P_1$ (Theorem ~\ref{thm:intersect}). In other words, the TSPN that visits the points $P_n - P_1 - P_2 - P_3$ can be reimagined as visiting the points $P_n - P_2 - Q_1 - P_3$ and therefore respecting the order $\sigma'$. The local cost of crossing these edges is the same as before: $|P_nP_2| + |P_2Q_1| + |Q_1P_3| = |P_nP_2| + |P_2P_3| = L_T$.

We now apply Theorem~\ref{thm:triangle} (the $3\sqrt{3}R$ bound for $n=3$) on the TSP tour $O_n -O_1 - O_2$ with the TSPN tour $P_n - P_1 - P_2$ and get that:
\begin{align*}
  |O_nO_1| + |O_1O_2| + |O_nO_2| &\leq |P_nP_1| + |P_1P_2| + |P_nP_2| + 3\sqrt{3}R &&\\
                          &\leq 2|P_2P_n| + 3\sqrt{3}R.&&
\end{align*}
On the other hand, if we consider the tour $O_1 - O_2- O_3$ with the TSPN tour $P_2 - Q_1 - P_3$, we get that:
\begin{align*}
 |O_1O_2| + |O_2O_3| + |O_1O_3| &\leq |Q_1P_2| + |P_2P_3| + |P_3Q_1| +  3\sqrt{3}R && \\
                          &\leq 2|P_2P_3| + 3\sqrt{3}R.&&
\end{align*}
Combining the two inequalities and rearranging some terms gives us that:
\begin{align*}
|TSP(\sigma)|+ |TSP(\sigma')| =& 2|T'| +  |O_nO_1| + |O_1O_2| + |O_2O_3| + \\
                               &+ |O_nO_2| + |O_2O_1| + |O_1O_3| \\
                              =& 2|T'|  + |O_nO_1| + |O_1O_2| + |O_nO_2| + \\ &+|O_1O_2| + |O_1O_3|+|O_2O_3| \\
         \leq & 2|T'|  + 2|P_2P_n| + 2 |P_2P_3| + 6\sqrt{3}R. 
\end{align*}
Since $L_T = |P_2P_n| + |P_2P_3|$, we get our conclusion.

When $k>1$, we construct the order $\sigma'$ by switching the order in which we visit the centers in each $\beta$-triad in the same way as before. Since all the $\beta$-triads are edge disjoint (Lemma ~\ref{lem:betatriad}), we can construct $\sigma'$ without any conflicts because any reordering that happens in one $\beta$-triad will not affect another $\beta$-triad.
\end{proof}

	\section{Improved Bounds on TSPN}
	\label{44}
	Our main strategy will be a careful balancing of good and bad edges, in which the detour of good edges will be upper bounded by $(1+\cos \beta)R$ and that of bad edges by $2R$. While the bad edges will have the highest detour possible, we will use the fact that they must also be large in order to lower bound the TSPN tour more efficiently than Lemma~\ref{lem:TSPNtour} from \cite{dumitrescu2003approximation} and \cite{dumitrescu2016traveling}, which we quote here for completeness. 

We quote the more general version formulated in \cite{dumitrescu2016traveling}, since we will actually use it with a slight modification. 

\begin{lemma}\cite{dumitrescu2016traveling}\label{lemma:slide} Given a connected geometric graph $G=(V,E)$ in $\mathbb{R}^2$ and $C$ the set of points that are at most $x$ away from the vertices and edges of $G$, we have that:
\begin{center}
 $ Area(C) \leq 2x \cdot |G| + \pi x^2$,
\end{center}
and this is tight in general.
\end{lemma}

The shape $C$ defined above can alternatively be thought of as the shape we describe when we slide a disk of radius $x$ along the edges of $G$ (i.e. the Minkowski sum of $G$ with a disk of radius $x$). The charging scheme behind the analysis starts at an arbitrary vertex of $G$ and initially pays a charge of $\pi x^2$ for it. Next, each edge $e$ we sweep from this vertex to the next will incur an additional charge of only $|e| \cdot 2x$. Because the graph $G$ is connected, we can continue this way and sweep through the entire graph while only incurring an additional charge of $e|G| \cdot 2x$. Next, we instantiate Lemma ~\ref{lemma:slide} with $G$ being a TSPN tour and $x = 2R$ and notice that the disk of radius $2R$ visiting a vertex $P$ on the boundary of a disk actually covers the entire disk of radius $R$ whose boundary $P$ is on. Since the disks are disjoint, we get that $Area(C) \geq \pi R^2 \cdot n$ and so we have that:

\begin{lemma}\cite{dumitrescu2003approximation,dumitrescu2016traveling}\label{lem:TSPNtour} For $n$ disjoint disks of radius $R$, we have that any TSPN tour $\mathcal{T}$ on them satisfies:
\begin{center}
    $ \frac{\pi}{4} Rn - \pi R \leq |\mathcal{T}|$.
\end{center}
\end{lemma}

%
%

\subsection{Disjoint Uniform Disks}  
We will now show the proof of Theorem~\ref{thm:mama}.
\begin{proof} Assume the  $TSPN^*$ is not a straight line. We start by singling out the $\beta$-triads and considering the two orderings $\sigma$ and $\sigma'$ from Theorem ~\ref{thm:triadbound}. If there are $k_1$ $\beta$-triads $\mathcal{T}_1, \ldots, \mathcal{T}_{k_1}$ spanning edges of total length $L_T$, we get that:
  \begin{align*}
   |TSP^*| &\leq \displaystyle{\frac{|TSP(\sigma)| + |TSP(\sigma')|}{2}} \\
            &\leq \displaystyle{ |TSP(\sigma \cap \sigma')| + L_T + 3\sqrt{3}R \cdot k_1}.
  \end{align*}
  
Observe that $TSPN(\sigma \cap \sigma')$ is a collection of disjoint paths. From all of these paths, we further extract each from these a total of $k_2$ subpaths $\mathcal{G}_1,\ldots, \mathcal{G}_{k_2}$ consisting of good edges. Notice that the remaining subpaths left in $\sigma \cap \sigma'$ consist of bad edges which do not form a $\beta$-triad.  Suppose we obtain $l$ such  remaining subpaths paths $\mathcal{B}_1,\ldots,  \mathcal{B}_l$. In other words, we have decomposed the TSPN into three categories of subpaths:
 \begin{itemize}
  \item $k_1$ $\beta$-triads  $\mathcal{T}_1, \ldots, \mathcal{T}_{k_1}$,
  \item $k_2$ paths $\mathcal{G}_1,\ldots, \mathcal{G}_{k_2}$ that cover the remaining good edges, and
  \item $l$ paths  $\mathcal{B}_1,\ldots,  \mathcal{B}_l$ that consist only of bad edges which do not form $\beta$-triads.
 \end{itemize}

We are now ready to evaluate the detour that each of these paths takes. For each $i \in [1,k_2]$ let $\psi_i$ the natural order on the disks associated with $\mathcal{G}_i$ and let $n_i$ be the number of edges in $\mathcal{G}_i$. We have that:
\begin{center}
 $|TSP(\psi_i)| \leq |TSPN(\psi_i)| + ( 1+\cos \beta)R \cdot n_i$.
\end{center}
 
When it comes to the paths $\mathcal{B}_j$, with $j \in [1,l]$, let $\sigma_j$ be their natural associated orders and let $m_j$ be the number of edges it contains. We have that $|TSP (\sigma_j)| \leq |TSPN(\sigma_j)| + 2R \cdot m_j$.
 
 Let $N = \sum_{i=1}^{k_2} n_i$ be the total number of edges in $\mathcal{G}_1,\ldots, \mathcal{G}_{k_2}$ and $M = \sum_{j=1}^{l} m_j$ the total number of edges in $\mathcal{B}_1,\ldots,  \mathcal{B}_l$. By construction, we decomposed $TSPN(\sigma \cap \sigma')$ into these two groups of edge disjoint paths and we therefore get that:
 \begin{align*}
   |TSP(\sigma \cap \sigma')| &= \displaystyle{\sum_{i=1}^{k_2}} |TSP(\psi_i)|  +  \displaystyle{\sum_{j=1}^{l}} |TSP(\sigma_j)| &&\\
                             &\leq \displaystyle{\sum_{i=1}^{k_2} \Big( |TSPN(\psi_i)| + (1 +\cos \beta) R \cdot n_i  \Big)}\\ 
                             & +  \displaystyle{\sum_{j=1}^{l} \Big( |TSPN(\sigma_j)| + 2R \cdot m_j \Big)} &&\\
                            &\leq |TSPN(\sigma \cap \sigma')| + (1 +\cos \beta)RN + 2R M. &&
 \end{align*}
 
 Including the $\beta$-triads back into our bound, we get that:
 \begin{align*}
 |TSP^*| &\leq \frac{|TSP(\sigma)| + |TSP(\sigma')|}{2}&& \\
         & \leq |TSP(\sigma \cap \sigma')| + L_T + 3\sqrt{3}R \cdot k_1 &&\\
         & \leq |TSPN(\sigma \cap \sigma')| + L_T +  3\sqrt{3}R \cdot k_1 +  (1 +\cos \beta)  RN + 2RM && \\
         & \leq |TSPN| + 3\sqrt{3}R \cdot k_1 +  (1 +\cos \beta)RN + 2RM. &&
 \end{align*}
 
In other words, we've expressed the total detour of the $TSPN$ according to  edges that participate in $\beta$-triads, edges in $\mathcal{G}_1,\ldots, \mathcal{G}_{k_2}$  and edges in  $\mathcal{B}_1,\ldots,  \mathcal{B}_l$. By construction, none of these paths share edges and so $3k_1 + N + M = n$. Let $K = 3k_1 + N$ be the total number of edges either in a $\beta$-triad or in $\mathcal{G}_1, \ldots, \mathcal{G}_{k_2}$ and since $\sqrt{3} \leq 1 +\cos \beta$, we have that:
\begin{center}
 $|TSP^*| \leq |TSPN| + (1 + \cos \beta)R \cdot K + 2R \cdot (n-K)$.
\end{center}

\textbf{Case $1$: when $K \geq \frac{n}{2}$.} In this situation, we have that:
\begin{center}
 $|TSP^*| \leq |TSPN^*| + \displaystyle{ \frac{3 +\cos \beta}{2}  }\cdot R \cdot n$. 
\end{center}

The average detour per edge $\frac{3+\cos \beta}{2}$ is better than the $2R$ bound, but it is constrained by the choice of $\beta \in [0, \pi/12]$, which means that the best we could hope for is an average detour of $\frac{1}{2}(3+\cos \frac{\pi}{12})R < 1.983R$. We note that the average detour in the H\"{a}me, Hyyti\"{a} and Hakula conjecture is $\sqrt{3}R \approx 1.732R$. Using Lemma ~\ref{lem:TSPNtour} gives us that
\begin{align*}
|TSP^*| &\leq  \displaystyle{ \Big(1 +  \frac{2}{\pi} \cdot(3+\cos \beta)\Big)} \cdot |TSPN^*| + \\
        &+ 2 \cdot (3+\cos \beta)R&.
\end{align*}
For large $n$, the $1 + \frac{2}{\pi} \cdot(3+\cos \beta)$ term will dominate our approximation factor and is at most $3.525$, when $\beta = \pi/12$. 

\textbf{Case $2$: when $K < \frac{n}{2}$.} In this situation, even the overall detour might be large, we will show that in fact, in this case, $TSP^*$ is a $2$-approximation and therefore, the best that it can be in general. We know that each path $\mathcal{B}_j$ consists of bad edges which do not form any $\beta$-triads. In other words, if $P_1P_2$ is an edge in it, then we know that $|O_1O2| > R/\sin(2\beta)$ which in turn means that $|P_1P_2| > (1/\sin(2\beta)-2) \cdot R$. Overall we have that:
\begin{align*}
&&|TSPN| &\geq \displaystyle{\sum_{j=1}^{l} |\mathcal{B}_j|} \geq \Big(\frac{1}{\sin{(2\beta})} -2\Big)R \cdot (n-K) &&\\
&&       &\geq \Big (\frac{1}{2\sin{(2\beta)}} -1 \Big)R \cdot n.&&
\end{align*}

Since the total detour could be at most $2R$ per edge, we get that:
\begin{center}
$|TSP^*| \leq \displaystyle{ \Big( 1 + \frac{2}{\frac{1}{2\sin{(2\beta)}} -1} \Big )} \cdot |TSPN|$.
\end{center}
When $\beta = \frac{1}{2}\arcsin\frac{1}{6}$, the detour from Case $1$ becomes $\frac{3 + \cos \beta}{2} \approx 1.998$ and the approximation factor from Case $2$ becomes exactly $2$. 
We note that the machinery described can be used to obtain more nuanced results. In particular, lower choices for $\beta$ will drive the approximation factor in Case $2$ even lower than $2$, at the expense of a higher detour bound for Case $1$. 

\end{proof}

More generally, we can consider a parameter $\alpha >1$ that we will set later in the proof. We include here only the aspects that change. Depending on whether $K \leq \frac{n}{\alpha}$ or not, we will employ different lower bounds on $|TSPN|$, in a similar fashion as before.

\textbf{Case $1$: when $K \geq \frac{n}{\alpha}$.} In this situation, we have that:
\begin{center}
 $|TSP^*| \leq |TSPN^*| + \displaystyle{ \frac{1 +\cos \beta + 2(\alpha-1)}{\alpha}  }\cdot R \cdot n$. 
\end{center}
 Using Lemma ~\ref{lem:TSPNtour} gives us that
\begin{align*}
|TSP^*| &\leq  \displaystyle{ \Big(1 +  \frac{4}{\pi} \cdot \frac{1 +\cos \beta + 2(\alpha-1)}{\alpha} \Big)} \cdot |TSPN^*| + 4 \cdot \frac{1 +\cos \beta + 2(\alpha-1)}{\alpha}R\\
         &\leq \displaystyle{ \Big(1 +  \frac{4}{\pi} \cdot \frac{1 +\cos \beta + 2(\alpha-1)}{\alpha} \Big)} \cdot |TSPN^*|+ 8R\\
         &\leq \displaystyle{ \Big(1 +  \frac{8}{\pi} - \frac{4}{\pi} \cdot \frac{1 -\cos \beta}{\alpha} \Big)} \cdot |TSPN^*|+ 8R.
\end{align*}

\textbf{Case $2$: when $K < \frac{n}{\alpha}$.}
 We know that each path $\mathcal{B}_j$ consists of bad edges which do not form any $\beta$-triads. In other words, if $P_1P_2$ is an edge in it, then we know that $|O_1O2| > R/\sin(2\beta)$ which in turn means that $|P_1P_2| > (1/\sin(2\beta)-2) \cdot R$. Overall we have that:
\begin{align*}
&&|TSPN| &\geq \displaystyle{\sum_{j=1}^{l} |\mathcal{B}_j|} &&\\
&&       &\geq \Big(\frac{1}{\sin{(2\beta})} -2\Big)R \cdot (n-K) &&\\
&&       &\geq \frac{\alpha-1}{\alpha} \cdot\Big(\frac{1}{\sin{(2\beta})} -2\Big)R \cdot n.&&
\end{align*}
Since the total detour could be at most $2R$ per edge, we get that:
\begin{center}
$|TSP^*| \leq \displaystyle{ \Big( 1 + \frac{\alpha}{\alpha-1} \cdot \frac{2}{\frac{1}{\sin{(2\beta)}}-2} \Big )} \cdot |TSPN|$.
\end{center}

If we want to achieve a factor $2$-approximation in Case $2$, we need to have $\beta \leq \frac{1}{2}\arcsin(\frac{1}{4}) $ and set
\begin{center}
 $\alpha = 1 + 2 / (\frac{1}{\sin{(2\beta)}}-4)$.
\end{center}
In this case, the detour in Case $1$ becomes $2-(1-\cos \beta)(2 - 1/(1-2\sin{(2\beta)})$ which achieves a minimum of $\approx 0.998$ on the interval $[0,\frac{1}{2}\arcsin(\frac{1}{4})]$. Setting $\alpha = 1 + 2/(c/{\sin(2\beta)} - 2c-2)$ for $c = 2.53$ and $\beta = 0.1831$ gives us that both of these cases lead to a $2.53$-approximation. 

\subsection{Overlapping Uniform Disks}

We discuss how the analysis from the disjoint case carries over to the case of overlapping disks. As we mentioned before, the best known approximation for this case is by Dumitrescu and T\'{o}th~\cite{dumitrescu2016traveling}. In general, approaches for this case take advantage of known analyses for the disjoint case and adapt them in a smart way to the overlapping case. We begin by roughly describing the technique of Dumitrescu and T\'{o}th~\cite{dumitrescu2016traveling} and the show how the analysis changes when we use our framework. 

Specifically, Dumitrescu and T\'{o}th~\cite{dumitrescu2016traveling} start by computing a monotone maximal set of disjoint disks $\mathcal{I}$ by greedily selecting the leftmost disk and deleting all of the other input disks that intersect it. Let $k$ be the size of the set we end up wth. They then compute an approximate TSP tour on the centers of the disks in $\mathcal{I}$, either using the available schemes~\cite{arora1998polynomial,mitchell1999guillotine} or Christofides~\cite{Chr76}. We call this tour $T_\mathcal{I}$. They then augment this tour in such a way that we visit all the input disks, not just the ones in $\mathcal{I}$. Before we discuss the augmentation part, we first define some notation and mention some bounds that follow naturally. 

Let the optimal TSP tour on the centers in $\mathcal{I}$ be $TSP^*_\mathcal{I}$. The eventual tour $T_\mathcal{I}$ that we compute will be an $a$-approximation to $TSPN^*_\mathcal{I}$ so we have that:
\begin{equation}\label{bound:1}
 |T_\mathcal{I}| \leq a \cdot |TSP^*_\mathcal{I}|. 
\end{equation}

On the other hand, we know that this set of disks also has an associated optimal TSPN tour, which we call $TSPN^*_{I}$. Finally, we denote the optimal TSPN tour on all the disks by $TSPN^*$. We know that the tour on $\mathcal{I}$ is a lower bound:
\begin{equation}\label{bound:2}
 |TSPN^*_\mathcal{I}| \leq |TSPN^*|. 
\end{equation}
The size of our final solution will be compared to $|TSPN^*|$ and to that end, we use lower bounds on $|TSPN^*_\mathcal{I}|$ in conjunction with (\ref{bound:2}) to get lower bounds on $|TSPN^*|$. This is the part where our new framework will come in, because $|TSPN^*_\mathcal{I}|$ is a tour on disjoint disks by definition. 

The next step is to augment $T_\mathcal{I}$ with detours of length $O(R)$ along the disks in $\mathcal{I}$ such that it touches every other disk not in $\mathcal{I}$. The total length of the solution would then become $|\mathcal{T}_{\mathcal{I}}| + O(1) \cdot |\mathcal{I}| \cdot R$. Specifically, Dumitrescu and T\'{o}th~\cite{dumitrescu2016traveling} consider short curves around each disk in $\mathcal{I}$ that are guaranteed to cross any of the disks to its right that intersect it. Because the maximal set was chosen from left to right, that covers all the disks that could possibly intersect it. We refer the reader to ~\cite{dumitrescu2016traveling} for the detailed construction. The authors show that the length of the resulting tour $T$ is within $O(1) \cdot |\mathcal{I}| \cdot R$ of $|T_\mathcal{I}|$:
\begin{equation}\label{bound:3}
 |T| \leq |T_\mathcal{I}| + (A\cdot k + B)\cdot R, 
\end{equation}
where $A = 2\cdot (\frac{\pi}{6} + \sqrt{3}-1)$ and $B = 4 - \sqrt{3}$. 

Combining \ref{bound:1} and \ref{bound:3}, we upper bound the length of the solution $|T|$ in terms of $ |TSP^*_\mathcal{I}|$ as such:
\begin{align*}
&&|T| & \leq  |T_\mathcal{I}| + (A\cdot k + B)\cdot R&&\\
&&    & \leq  a \cdot |TSP^*_\mathcal{I}| + (A\cdot k + B)\cdot R &&
\end{align*}

In order to complete the analysis, we would need to bound $ |TSP^*_\mathcal{I}|$ in terms of $|TSPN^*|$ and we do that through $|TSPN^*_{\mathcal{I}}|$. The analysis from  Dumitrescu and T\'{o}th~\cite{dumitrescu2016traveling} uses the bounds from Dumitrescu and Mitchell~\cite{dumitrescu2003approximation} for the case of disjoint disks. Specifically, they apply Lemma~\ref{lem:TSPNtour} to get that:
\begin{equation*}
 kR \leq \frac{4}{\pi} \cdot |TSPN^*_{\mathcal{I}}| + 4 R. 
\end{equation*}
This, together with the bound $|TSP^*_{\mathcal{I}}| \leq |TSPN^*_{\mathcal{I}}| + 2Rk$ and (\ref{bound:2}) yields:

 \begin{align*}
 &&|T| & \leq  a\cdot |TSP^*_\mathcal{I}| + (A k + B)\cdot R&&\\
 &&    & \leq  a \cdot (|TSPN^*_{\mathcal{I}}| + 2Rk) + (A k + B)\cdot R&&\\
 &&    &\leq  a \cdot |TSPN^*_{\mathcal{I}}| + (2 a + A) \cdot kR + BR&&\\
 &&    &\leq a \cdot |TSPN^*_{\mathcal{I}}| + (2 a+ A) \cdot \Big( \frac{4}{\pi} |TSPN^*_{\mathcal{I}}| + 4 R\Big)  + BR &&\\
 &&    &\leq \Big(a + (2 a + A)\frac{4}{\pi}\Big) \cdot |TSPN^*_{\mathcal{I}}| + (8a + 4 A + B)R&&\\
 &&   &\leq \Big((1+\frac{8}{\pi})a+\frac{4A}{\pi}\Big) \cdot |TSPN^*_{\mathcal{I}}| + (8a + 4 A + B)R&&\\
 &&   &\leq \Big((1+\frac{8}{\pi})a+\frac{4A}{\pi}\Big) \cdot |TSPN^*| + (8a + 4 A + B)R&&
 \end{align*}

Plugging in the values for $A$ and $B$ gives an overall approximation term of:
\begin{equation*}
 (1+\frac{8}{\pi})a+\frac{4A}{\pi} \leq \Big( \frac{7}{3} + \frac{8\sqrt{3}}{\pi}\Big) \cdot (1+\epsilon) \leq 6.75\cdot(1+\epsilon).
\end{equation*}

Our framework changes the last stage in which we compare $|TSP^*_{\mathcal{I}}|$ with $|TSPN^*_{\mathcal{I}}|$. We do a similar analysis as in the disjoint case, except for the tour on $\mathcal{I}$. We get that \textbf{Case 1} would therefore correspond to getting that:
\begin{equation*}
 |TSP^*_\mathcal{I}| \leq |TSPN^*_\mathcal{I}| + X \cdot R \cdot k, 
\end{equation*}
where $ X = \displaystyle{ 2 - \frac{1 -\cos \beta }{\alpha}  }$ (instead of $2R$).
We can then replace it in the analysis and get:
 \begin{align*}
 && |T| & \leq   a \cdot (|TSPN^*_{\mathcal{I}}| + XRk) + (A k + B)\cdot R&&\\
 &&  &\leq \Big(a + (X a + A)\frac{4}{\pi}\Big) \cdot |TSPN^*_{\mathcal{I}}| + (8a + 4 A + B)R&&\\
 &&  & \leq \Big((1+\frac{4X}{\pi})a+\frac{4A}{\pi}\Big) \cdot |TSPN^*| + (4Xa + 4 A + B)R  &&
 \end{align*}

In \textbf{Case $2$}, we have that the overall detour is $2Rk$, but there is a different lower bound on $|TSPN^*_\mathcal{I}|$:
\begin{equation*}
|TSPN^*_\mathcal{i}| \geq Y \cdot Rk,
\end{equation*}
where  $Y = \frac{\alpha-1}{\alpha} \cdot \Big (1/(2\sin{(2\beta)}) -1 \Big)$. Using the fact that $Rk \leq 1/Y \cdot |TSPN^*_\mathcal{I}|$, the analysis then becomes:
\begin{align*}
 && |T|   &\leq  \alpha \cdot |TSPN^*_{\mathcal{I}}| + (2 \alpha + A) \cdot kR + BR&&\\
 &&       &\leq  \alpha \cdot |TSPN^*_{\mathcal{I}}| + \frac{2 \alpha + A}{Y} \cdot |TSPN^*_{\mathcal{I}}| + BR&&\\
 &&       &\leq \Big( \alpha + \frac{2 \alpha + A}{Y} \Big) \cdot |TSPN^*_{\mathcal{I}}|  + BR.&&
\end{align*}
If we set $\alpha $ and $\beta$ like in the previous section, we get that both of the approximation factors are upper bounded by $6.728$.

\subsection{The Straight Line Case}\label{app:line}

Here we focus on the second possibility in Theorem ~\ref{thm:main} in which the optimal TSPN is supported by a straight line that stabs all the disks. We show that in this case, we can return in polynomial time a solution that is within an additive factor of $4R$ from the optimal $TSPN^*$. We note that when the TSPN might not be a line but the disks themselves admit a line transversal, a $\sqrt{2}$-approximation follows from the work of Dumitrescu and Mitchell~\cite{dumitrescu2003approximation}. We explain the result for completeness.

We start by identifying the centers that are the farthest apart and considering the direction orthogonal to the line going through them. This direction induces parallel segments of length $2R$ in each of the disks (that each go through the centers). It is easy to check that any line transversal through the disks is a line transversal through the segments except for the first and last disk in the associated geometric permutation (for those two disks, the TSPN will stop at the boundary of the disk and never cross the entire circle). Conversely, any line transversal through the segments will automatically also stab the disks. Now compute a shortest line segment that stabs all of these segments in time $O(n\log n)$ using the algorithm of Bhattacharya et al.~\cite{bhattacharya1992computing}. We note that this is optimal up to an additive factor of $4R$ that comes from the fact that the optimal $TSPN^*$ might have to travel $4R$ to hit the first and the last two segments in the geometric permutation. 

In general, when we know that the disks admit a line transversal, we can output a solution that is a $\sqrt{2}$-approximation~\cite{dumitrescu2003approximation}. This follows indirectly from an algorithm used for connected regions of the same diameter, when there is a line that stabs all of the diameters. Given the parallel segments of length $2R$ that we constructed earlier, we know that they can also be stabbed by a line. Now consider the smallest perimeter axis-aligned rectangle that intersects all of the segments, of width $w$ and height $h$. This will be the solution that we return. Arkin and Hassin~\cite{arkin1994approximation} argued that any tour which touches all four sides of the rectangle must have length at least $2\sqrt{h^2+w^2}$. Since $h+w \leq \sqrt{2}\cdot \sqrt{h^2+w^2}$, we get that the rectangle is a $\sqrt{2}$-approximation.

	\section{The Fermat-Weber Approach}
	\label{45}
	\label{fermat}

In this section, we prove that the H\"{a}me, Hyyti\"{a} and Hakula conjecture is true for $n=3$ and discuss a different way of looking at the TSPN tour that we believe might be of independent interest. We start with the observation that the shortest tour on the centers is equivalent to the shortest tour on translates of those centers, as long as all those centers are translated according to the same vector. In other words, if we fix a direction and translate each center along that direction until it reaches its boundary, the shortest tour on the newly obtained points will be exactly the same as the shortest tour on the centers themselves.

\begin{figure}[h]
	\centering
	\renewcommand{\baselinestretch}{1}
	\small\normalsize
	\includegraphics[height =2in]{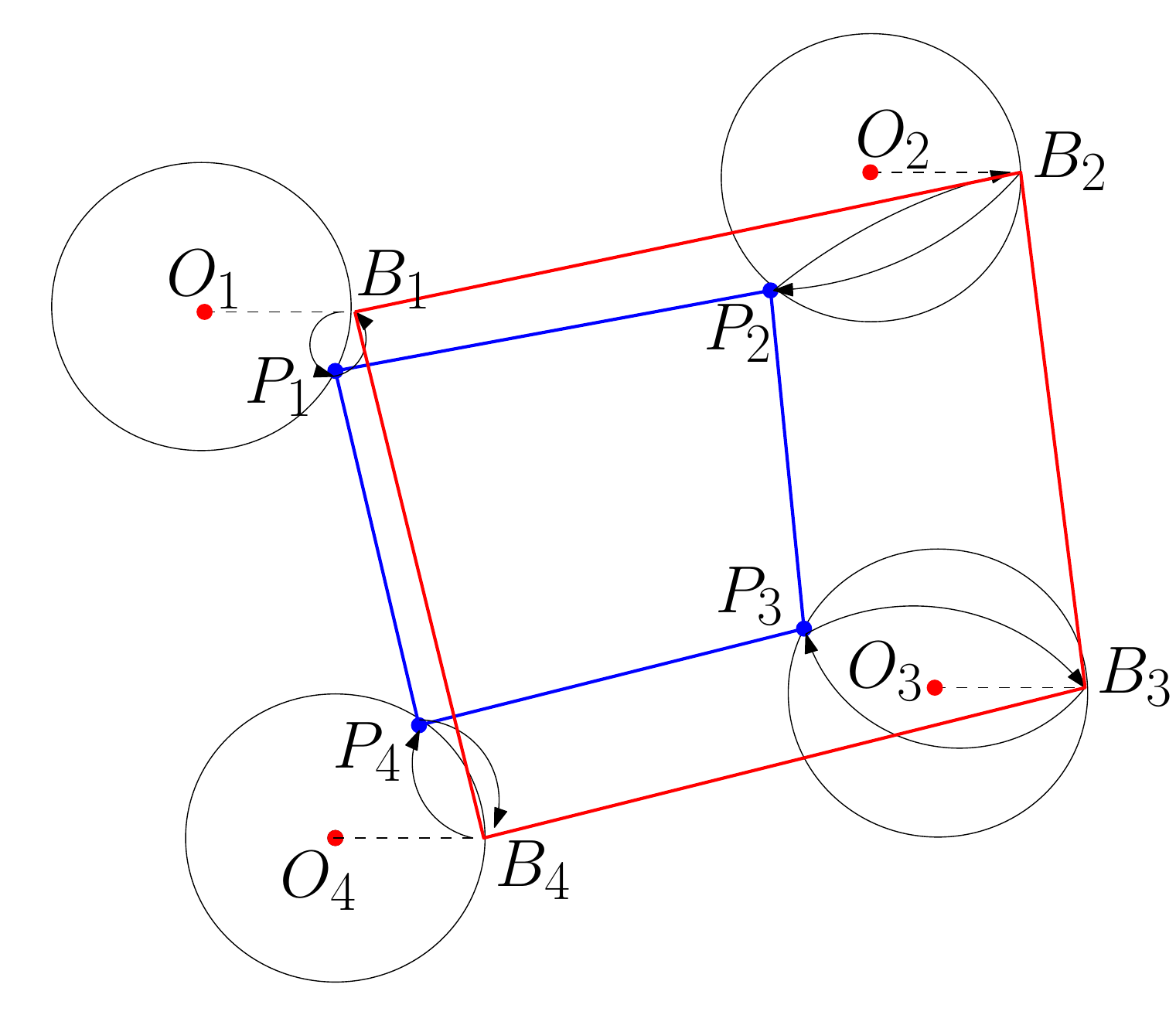}
	\caption[The translated view of the tour.]{The translated view, when the tour visits the same point on the boundary of each disk.}
	\label{fig:sub3}
\end{figure}


%

Formally, let $B_i$ be the point we obtain by translating the center $O_i$ along a fixed vector of length $R$. Then the TSP on the points $B_1,B_2,\ldots, B_n$ (in that order) has the same length as the TSP tour on $O_1, O_2, \ldots O_n$ (Figure~\ref{fig:sub3}). One advantage of visiting the first set of points (instead of the center points) is that it might be more similar geometrically to what the TSPN actually does. In terms of the following analysis, we would get that:
\begin{center}
 $|TSP^*| \leq |TSPN^*| + 2\cdot \displaystyle{\sum_{i=1}^{n}} |P_iB_i|$.
\end{center}

In this context, a natural question arises about the choice for the points $B_i$ that minimizes the term $\sum_{i=1}^{n} |P_iB_i|$. In order to see what this best choice would be, we transform this input instance into another one by essentially superimposing all the disks on top of each other (Figure~\ref{fig:fermat}(a)). Specifically, our new instance will consist of one disk of radius $R$ centered at a point $O$ such that the points $B_i$  map to a single point $B$ (corresponding to $O$ translated by the same fixed vector). We then map each point $P_i$ of the TSPN to a corresponding point $Q_i$ on the boundary of this disk such the vector $OQ_i$ is a translate of the vector $O_iP_i$. We then get that:
\begin{center}
 $\displaystyle{\sum_{i=1}^{n}} |P_iB_i| = \displaystyle{\sum_{i=1}^{n}} |Q_iB|$,
\end{center}
and so the best choice for $B$ is the one that minimizes the sum $\sum_{i=1}^{n} |Q_iB|$, otherwise know as the \textit{Fermat-Weber point} or \textit{$1$-median} of the points $Q_1, Q_2, \ldots, Q_n$~\cite{wesolowsky1993weber, weber1909ueber}. We note, however, that while the average distance to the Fermat-Weber point will never be greater than $2R$, there are instances in which this is tight. Consider, for example, the points $Q_i$ to be the vertices of a convex $2n$-gon and notice by triangle inequality that the center of the disk is exactly their Fermat-Weber point (any other point will incur distances greater than the sum of the diagonals). 

\begin{figure}
	\centering
	\renewcommand{\baselinestretch}{1}
	\small\normalsize
	\subfigure[Unified view when we translate each $O_i$ to the same point on the boundary.]{\includegraphics[height = 2in]{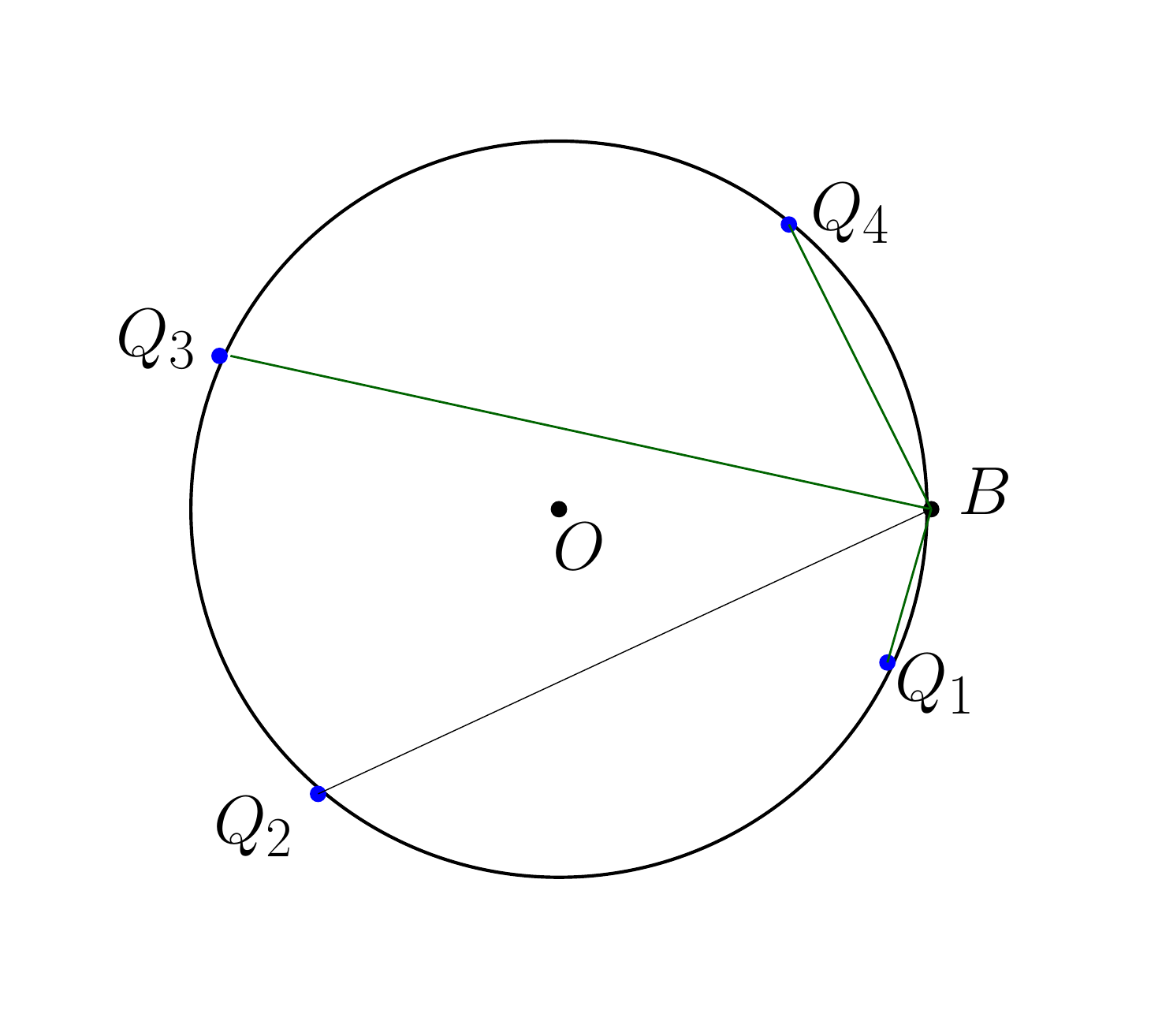}} 
	\hspace{1cm}
	\subfigure[$n$ different choices for $B$, when we choose a different $B$ for each pairs of points $O_i$ and $O_{i+1}$]{\includegraphics[height = 2in]{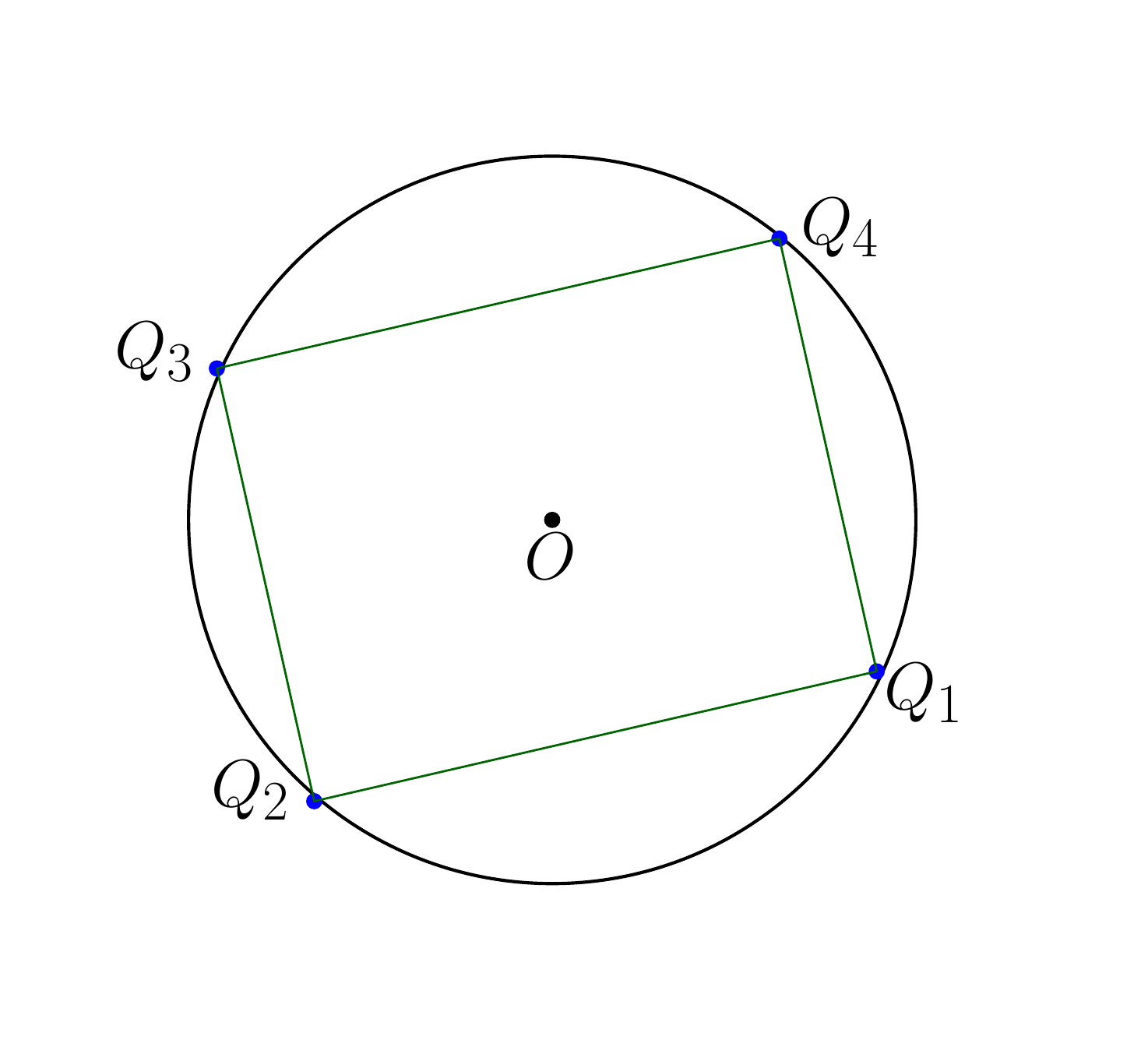}}
	\caption[Using Fermat-Weber points to improve the detour bound. ]{The detour bound depending on what kind of Fermat-Weber points we consider.}
	\label{fig:fermat}
\end{figure}

We can therefore say that when the points $B_i$ are evenly spaced on the boundary of the disk the Fermat-Weber point is exactly the center and so we gain no improvement by moving the centers $O_i$ towards the points $B_i$. It turns out, however, that the location of the points on the boundary is not as restrictive as the order in which the TSPN visits them. To see that, consider a different transformation in which we only move the centers $O_i$ and $O_{i+1}$ along a fixed vector. In other words, we choose a new vector for each pair of consecutive centers and only compare $|P_iP_{i+1}|$ locally against the newly obtained segment. This does not give us an overall valid tour on the centers, but it allows us to tailor the choice of $B$ for each two points $P_i$ and $P_{i+1}$. Specifically, we would get that:
\begin{center}
$|O_i O_{i+1}| \leq |P_iP_i+1| + |Q_iB| + |Q_{i+1}B|$.
\end{center}
In this case, we know that any point on the segment $Q_iQ_{i+1}$ minimizes the distances in question and so we get that:
\begin{center}
$|O_i O_{i+1}| \leq |P_iP_i+1| + |Q_iQ_{i+1}|$  and $|TSP^*| \leq |TSPN^*| + \displaystyle{\sum_{i=1}^n} |Q_iQ_{i+1}|$.
\end{center}
In other words, the largest detour obtained in this way is when the TSPN visits the points $P_i$ in the order of the Maximum TSP on the associated points $Q_i$ (Figure~\ref{fig:fermat}(b)). The case in which all the points are evenly distributed along the boundary no longer becomes that restrictive. We can still construct, however, instances for which the Max TSP is exactly $2Rn$ and that is when the points visited are exactly diametrically opposite each other. Nevertheless, we are able to show that for $n=3$, the detour is bounded by $3\sqrt{3}R$. Let $A, B, C$ be any three points on the boundary of a circle of disk $R$ centered at $O$. We then have that that  $|AB| + |AC| + |BC| \leq  3\sqrt{3}R$ and the H\"{a}me, Hyyti\"{a} and Hakula conjecture for $n=3$ follows:
\begin{theorem}\label{thm:triangle} For $n=3$, we have that any tour which visits the disks in an order $\sigma$ satisfies the bound
\begin{center}
$|TSP(\sigma)| \leq |TSPN(\sigma)| + 3\sqrt{3}R$.
\end{center}
\end{theorem}

	\subparagraph*{Acknowledgements.}
	
	The author would like to thank Prof. Samir Khuller for suggesting the problem and for inspiring conversations on the topic.


	
	\small
	\bibliographystyle{abbrv}
	\bibliography{abbrv}

\end{document}